\documentclass[a4paper,USenglish,cleveref, thm-restate]{lipics-v2019}

\author{Dominik K\"{o}ppl}{Department of Informatics, Kyushu University \and Japan Society for Promotion of Science (JSPS), Japan}{dominik.koeppl@inf.kyushu-u.ac.jp}{https://orcid.org/0000-0002-8721-4444}{JSPS KAKENHI Grant Number JP18F18120.}
\author{Daiki Hashimoto}{Graduate School of Information Sciences, Tohoku University, Japan}{daiki\_hashimoto@shino.ecei.tohoku.ac.jp}{}{}
\author{Diptarama Hendrian}{Graduate School of Information Sciences, Tohoku University, Japan}{diptarama@tohoku.ac.jp}{https://orcid.org/0000-0002-8168-7312}{JSPS KAKENHI Grant Number JP19K20208.}
\author{Ayumi Shinohara}{Graduate School of Information Sciences, Tohoku University, Japan}{ayumis@tohoku.ac.jp}{https://orcid.org/0000-0002-4978-8316}{JSPS KAKENHI Grant Number JP15H05706.}

\authorrunning{D. K\"{o}ppl, D. Hashimoto, D. Hendrian, and A. Shinohara} 

\Copyright{Dominik K\"{o}ppl, Daiki Hashimoto, Diptarama Hendrian, and Ayumi Shinohara} 



\category{} 

\relatedversion{\href{https://doi.org/10.4230/LIPIcs.CPM.2020.23}{10.4230/LIPIcs.CPM.2020.23}}

\supplement{At \url{https://github.com/daikihashimoto/BWT_to_BBWT},
we have some preliminary implementations available 
giving empirical evidence of our conversions.}

\nolinenumbers{} 

\hideLIPIcs{}

\usepackage{tikz} 
\usetikzlibrary{arrows.meta}

\newcommand*{\IroHako}[1]{\resizebox{1em}{!}{\begin{tikzpicture}[inner sep=0pt, outer sep=0pt, baseline]
		\draw [draw=none] (0,0) rectangle (1,1);
	\draw[line width=0.5em, #1] (0,0.5) -- (1,0.5);
\end{tikzpicture}
}
\hspace{-0.7em}
}

\usepackage{float} 

	\Crefname{algocf}{Algorithm}{Algorithms}
	\crefname{algocf}{Algo.}{Algos.}
	\Crefname{equation}{Equation}{Equations}
	\crefname{equation}{Eq.}{Eqs.}
    \Crefname{figure}{Figure}{Figures}\crefname{figure}{Fig.}{Figs.}\Crefname{theorem}{Theorem}{Theorems}
	\crefname{theorem}{Thm.}{Thms.}

	\Crefname{definition}{Definition}{Definitions}
	\crefname{definition}{Def.}{Defs.}
	\Crefname{corollary}{Corollary}{Corollaries}
	\crefname{corollary}{Cor.}{Cors.}
	\Crefname{section}{Section}{Sections}
	\crefname{section}{Sect.}{Sects.}

\usepackage{xcolor}
	\definecolor{solarizedYellow}{HTML}{B58900}
	\definecolor{solarizedOrange}{HTML}{CB4B16}
	\definecolor{solarizedRed}{HTML}{DC322F}
	\definecolor{solarizedMagenta}{HTML}{D33682}
	\definecolor{solarizedViolet}{HTML}{6C71C4}
	\definecolor{solarizedBlue}{HTML}{268BD2}
	\definecolor{solarizedCyan}{HTML}{2AA198}
	\definecolor{solarizedGreen}{HTML}{859900}

\usepackage{amsmath}
\usepackage{hyperref}
\usepackage{booktabs}
\usepackage{diagbox}
\usepackage{adjustbox}

\DeclareMathAlphabet{\mathup}{OT1}{\familydefault}{m}{n}

\newcommand{\UnaryOperator}[2][]{\ifx&#1&\ensuremath{\mathop{}\mathopen{}#2\mathopen{}}\else \ensuremath{\mathop{}\mathopen{}#2\mathopen{}(#1)}\fi }

\newcommand{\Oh}[1]{\UnaryOperator[#1]{\mathcal{O}}}
\newcommand{\oh}[1]{\UnaryOperator[#1]{o}}

\newcommand{\Om}[1]{\UnaryOperator[#1]{\mathup{\Omega}}}

\newcommand*{\twodots}{\mathbin{{.}\,{.}}}

\usepackage[ruled,linesnumbered,algo2e,vlined]{algorithm2e}
\SetKwComment{Comment}{$\triangleright$\ }{}

\newcommand*{\instancename}[1]{\ensuremath{\mathsf{#1}}} \newcommand*{\functionname}[1]{{\ensuremath{\renewcommand{\rmdefault}{ptm}\fontfamily{ppl}\selectfont\textrm{\textup{#1}}}}} 

\definecolor{teigiColor}{HTML}{5700B5}
\newcommand*{\teigi}[1]{\emph{\color{teigiColor}#1}}

\newcommand*{\RLBBWT}     {\instancename{RLBBWT}}
\newcommand*{\RLBWT}     {\instancename{RLBWT}}
\newcommand*{\BBWT}     {\instancename{BBWT}}
\newcommand*{\BWT}     {\instancename{BWT}}
\newcommand*{\BWTC}     {\ensuremath{\instancename{BWT}^\circ}}
\newcommand*{\SA}     {\instancename{SA}}

\newcommand*{\OmegaOrder}{\ensuremath{\prec_\omega}}
\newcommand*{\First}     {\instancename{F}}
\newcommand*{\LF}     {\instancename{LF}}
\newcommand*{\FL}     {\instancename{FL}}

\let\citet\cite

\newcommand*{\conj}[2][]{\ensuremath{\functionname{conj}_{#1}(#2)}}

\newcommand*{\runText}{\ensuremath{r_{T}}}
\newcommand*{\runBWT}{\ensuremath{r_{\BWT}}}
\newcommand*{\runBBWT}{\ensuremath{r_{\BBWT}}}

\newcommand*{\LexOrder}{\ensuremath{\prec_{\textup{lex}}}}
\newcommand*{\LexOrderEq}{\ensuremath{\preceq_{\textup{lex}}}}
\newcommand*{\LexOrderReq}{\ensuremath{\succeq_{\textup{lex}}}}
\newcommand*{\RunningExample}{bacabbabb}

\newcommand*{\ibeg}[1]{\ensuremath{\functionname{b}(#1)}}
\newcommand*{\iend}[1]{\ensuremath{\functionname{e}(#1)}}
\newcommand*{\rank}{\functionname{rank}}
\newcommand*{\select}{\functionname{select}}

\newcommand*{\ImgWidth}{\linewidth}

\title{In-Place Bijective Burrows-Wheeler Transforms}

\begin{document}
	
	\maketitle
	
\begin{abstract}
	One of the most well-known variants of the Burrows-Wheeler transform (BWT) [Burrows and Wheeler, 1994] is the bijective BWT (BBWT) [Gil and Scott, arXiv 2012], 
	which applies the extended BWT (EBWT) [Mantaci et al., TCS 2007] to the multiset of Lyndon factors of a given text.
	Since the EBWT is invertible,
	the BBWT is a bijective transform in the sense that the inverse image of the EBWT restores this multiset of Lyndon factors such that the original text can be obtained by sorting these factors in non-increasing order.

In this paper, we present algorithms constructing or inverting the BBWT in-place using quadratic time.
We also present conversions from the BBWT to the BWT, or vice versa, either (a) in-place using quadratic time, 
or (b) in the run-length compressed setting using \Oh{n \lg r / \lg \lg r} time with \Oh{r \lg n} bits of words, 
where $r$ is the sum of character runs in the BWT and the BBWT.
\end{abstract}

\clearpage

\section{Introduction}
The \emph{Burrows-Wheeler transform} (BWT)~\cite{burrows94bwt} is one of the most favored options
both for (a)~compressing and (b)~indexing data sets.
On the one hand, compression programs like \texttt{bzip2} apply the BWT to achieve high compression rates.
For that, they leverage the effect that the BWT built on repetitive data tends to have long character runs, 
which can be compressed by run-length compression, i.e., representing a substring of $\ell$ \texttt{a}'s by the tuple $(\texttt{a},\ell)$.
On the other hand, self-indexing data structures like the FM-index~\cite{ferragina00fmindex} enhance the BWT to a full-text self-index.
A combined approach of both compression and indexing is the run-length compressed FM-index~\cite{makinen05rle}, 
representing a BWT with $\runBWT$ character runs, i.e., maximal repetitions of a character, run-length compressed in $\Oh{\runBWT \lg n}$ bits.
This representation can be computed directly in run-length compressed space thanks to {Policriti and Prezza}~\citet{policriti18lz77}.
The BWT and its run-length compressed representation have been intensively studied during the past decades (e.g., \cite{ferragina04bwt,adjeroh08bwt,gagie19bwtdagstuhl} and the references therein).
Contrary to that, a variant, called the \emph{bijective} BWT (BBWT)~\cite{gil12bbwt}, is far from being well-studied despite its mathematically appealing characteristics\footnote{The BBWT is a bijection between strings without the need of an artificial delimiter needed, e.g., to invert the BWT.}.
As a matter of fact, we are only aware of one index data structure based on the BBWT~\cite{bannai19bbwt} and of two non-trivial construction algorithms~\cite{bonomo14sorting,bannai19bbwtconstarxiv} of the (uncompressed) BBWT, both with the need of additional data structures.

In this article, we shed more light on the connection between the BWT and the BBWT by quadratic time in-place conversion algorithms in \cref{secInPlaceConversions} constructing the BWT from the BBWT, or vice versa.
We can also perform these conversions in the run-length compressed setting 
in \Oh{n \lg r /\lg \lg r} time with space linear to the number of the character runs (cf.~\cref{secRunLengthCompressed} and \cref{thmRunLength}), 
where $r$ is the sum of character runs in the BWT and the BBWT.

\section{Related Work}
Given a text~$T$ of length~$n$,
the BWT of~$T$ is the string obtained by assigning $\BWT[i]$ to the character preceding the $i$-th lexicographically smallest suffix of~$T$ (or the last character of~$T$ if this suffix is the text itself).
By this definition, we can construct the BWT with any suffix array~\cite{manber93sa} construction algorithm.
However, storing the suffix array inherently needs $n \lg n$ bits of space.
{Crochemore et al.}~\citet{crochemore15bwt} tackled this space problem with an in-place algorithm constructing the BWT in \Oh{n^2} online on the reversed text
by simulating queries on a dynamic wavelet tree~\cite{grossi03wavelet} that would be built on the (growing) BWT\@.
They also gave an algorithm for restoring the text in-place in \Oh{n^{2+\epsilon}} time.

In the run-length compressed setting,
{Policriti and Prezza}~\citet{policriti18lz77} can compute the run-length compressed BWT having $\runBWT$ character runs
in \Oh{n \lg \runBWT} time while using \Oh{\runBWT \lg n} bits of space.
They additionally presented an adaption of the wavelet tree on run-length compressed texts, 
yielding a representation using \Oh{\runBWT \lg n} bits of space with \Oh{\lg \runBWT} query and update time.
Finally, practical improvements of the run-length compressed BWT construction were considered by {Ohno et al.}~\citet{ohno17rlbwt}.

The BBWT is the string obtained by assigning $\BBWT[i]$ to the last character of the $i$-th smallest string 
in the list of all conjugates of the factors of the Lyndon factorization sorted with respect to the $\OmegaOrder$ order~\cite[Def.~4]{mantaci07ebwt}.
{Bannai et al.}~\citet{bannai19bbwtconstarxiv} recently revealed a connection between the bijective BWT and suffix sorting by
presenting an \Oh{n} time BBWT construction algorithm based on SAIS~\cite{nong11sais}.
With dynamic data structures like a dynamic wavelet tree~\cite{navarro14dynamic}, 
{Bonomo et al.}~\citet{bonomo14sorting} could devise an algorithm computing the BBWT in \Oh{n \lg n / \lg \lg n} time.
With nearly the same techniques, Mantaci et al.~\cite{mantaci14sa} presented an algorithm computing the BWT (and simultaneously the suffix array if needed)
from the Lyndon factorization. 
All these construction algorithms need however data structures taking \Oh{n \lg n} bits of space.
However, the latter two (i.e.,~\cite{bonomo14sorting} and~\cite{mantaci14sa}) can work in-place by simulating the \LF{} mapping (cf.~\cref{secSearchSteps,secRewinding}), 
which we focus on in \cref{secConstructBWTC}.

\section{Preliminaries}
Our computational model is the word RAM model with word size \Om{\lg n}.
Accessing a word costs \Oh{1} time.
An algorithm is called \teigi{in-place} if it uses, besides a rewriteable input, only \Oh{\lg n} bits of working space.
We write $[\ibeg{I}\twodots{}\iend{I}] = I$ for an interval~$I$ of natural numbers.

\subsection{Strings}
Let $\Sigma$ denote an integer alphabet of size~$\sigma$ with $\sigma = n^{\Oh{1}}$.
We call an element $T \in \Sigma^*$ a \teigi{string}.
Its length is denoted by $|T|$.
Given an integer~$j \in [1\twodots{}|T|]$, we access the $j$-th character of~$T$ with~$T[j]$.
Concatenating a string $T \in \Sigma^*$ $k$ times is abbreviated by $T^k$.
A string~$T$ is called \teigi{primitive} if there is no string~$S \in \Sigma^+$ with $T = S^k$ for an integer~$k$ with $k \ge 2$.

When $T$ is represented by the concatenation of $X, Y, Z \in \Sigma^*$, i.e., $T = \textit{XYZ}$,
then $X$, $Y$ and $Z$ are called a \teigi{prefix}, \teigi{substring} and \teigi{suffix} of $T$, respectively;
the prefix $X$, substring~$Y$, or suffix~$Z$ is called \teigi{proper} if $X \neq T$, $Y \neq T$, or $Z \neq T$, respectively.
For two integers~$i,j$ with $1 \leq i \leq j \leq |T|$, let $T[i\twodots{}j]$ denote
the substring of $T$ that begins at position~$i$ and ends at position~$j$ in~$T$.
If $i > j$, then $T[i\twodots{}j]$ is the empty string.
In particular, the suffix starting at position~$j$ of~$T$ is called the \teigi{$j$-th suffix} of~$T$, and denoted with $T[j\twodots{}]$.
An occurrence of a substring~$S$ in $T$ is treated as a sub-interval of $[1\twodots{}|T|]$ such that $S = T[\ibeg{S}\twodots{}\iend{S}]$. 
The \teigi{longest common prefix (LCP)} of two strings~$S$ and~$T$ is the longest string that is a prefix of both $S$ and~$T$.

\subparagraph{Orders on Strings.}
We denote the \teigi{lexicographic order} with $\LexOrder{}$.
Given two strings~$S$ and~$T$, $S \LexOrder T$ if $S$ is a prefix of $T$ or there exists an integer~$\ell$ with
$1 \le \ell \le \min(|S|,|T|)$ such that $S[1\twodots{}\ell-1]=T[1\twodots{}\ell-1]$ and $S[\ell] < T[\ell]$.
Next we define the \teigi{$\OmegaOrder{}$ order} of strings, which is based on the lexicographic order of infinite strings:
We write $S \OmegaOrder{} T$ if the infinite concatenation $S^\omega := {SSS} \cdots$ is lexicographically smaller than 
$T^\omega := {TTT}\cdots$. 
For instance, $\texttt{ab} \LexOrder \texttt{aba}$ but $\texttt{aba} \OmegaOrder \texttt{ab}$.

\subparagraph{Rank and Select Queries.}
Given a string~$T \in \Sigma^*$, 
a character $\texttt{c} \in \Sigma$, and an integer~$j$, 
the \teigi{rank} query $T.\rank_{\texttt{c}}(j)$ counts the occurrences of \texttt{c} in $T[1\twodots{}j]$, and 
the \teigi{select} query $T.\select_{\texttt{c}}(j)$ gives the position of the $j$-th \texttt{c} in $T$.
We stipulate that $\rank_{\texttt{c}}(0) = \select_{\texttt{c}}(0) = 0$.
A \teigi{wavelet tree} is a data structure supporting rank and select queries.

\subsection{Lyndon Words}
Given a string~$T = T[1\twodots{}n]$, its $i$-th \teigi{conjugate} \conj[i]{T} is defined as $T[i+1\twodots{}n]T[1\twodots{}i]$ for an integer $i \in [0\twodots{}n-1]$.
We say that $T$ and all of its conjugates belong to the \teigi{conjugate class} $\conj{T} := \{ \conj[0]{T}, \ldots, \conj[n-1]{T} \}$.
If a conjugate class contains \emph{exactly} one conjugate that is lexicographically smaller than all other conjugates,
then this conjugate is called a \teigi{Lyndon word}~\cite{lyndon54}.
Equivalently, a string $T$ is said to be a Lyndon word if and only if $T \LexOrder S$ for every proper suffix $S$ of $T$~\cite[Prop.~1.2]{duval83lyndon}.

The \teigi{Lyndon factorization}~\cite{chen58lyndon} of $T\in\Sigma^+$ 
is the factorization of $T$ into a sequence of lexicographically non-increasing Lyndon words~$T_1 \cdots T_t$,
where (a) each $T_x\in\Sigma^+$ is a Lyndon word for $x \in [1\twodots{}t]$, and (b) $T_x \LexOrderReq T_{x+1}$ for each $x \in [1\twodots{}t)$.
Each Lyndon word $T_x$ is called a \teigi{Lyndon factor}.

\begin{lemma}[{\cite[Algo.\ 2.1]{duval83lyndon}}]\label{lemLyndonFactorizationLinearTime}
	Given a string~$T$ of length~$n$,
	there is an algorithm that outputs  
the Lyndon factors~$T_1, \ldots, T_t$ one by one in increasing order in \Oh{n} total time
	while keeping only a constant number of pointers to positions in~$T$ that (a) can move one position forward at one time or (b) can be set to the position of another pointer. 
\end{lemma}
\begin{proof}
	The algorithm of Duval uses three variables $i$, $j$, and $k$ (cf.~\cref{algoDuval} in the appendix) pointing to text positions.
	$k$ is the ending position of the previously computed Lyndon factor (or zero at the beginning).
	On each step, $j \in [k+2\twodots{}n]$ is incremented by one, while $i$ is either incremented by one or reset to $k+1$, 
	as long as $T[k+1\twodots{}i-1] = T[j-(i-k)\twodots{}j-1]$ is a prefix of a Lyndon word starting at $T[k+1\twodots{}]$.
	If $T[k+1\twodots{}i]$ is no longer such a prefix, 
	then $T[k+1\twodots{}i-1]$ is either a Lyndon factor or a repetition of Lyndon factors, each of length $j-i$.
	In total, we visit at most $2n$ characters by incrementing the text positions~$i$, $j$, and~$k$.
\end{proof}

For what follows, we fix a string $T[1\twodots{}n]$ over an alphabet $\Sigma$ with size~$\sigma$.
We use the string $T := \texttt{\RunningExample}$ as our running example.
Its Lyndon factors are $T_1 = \texttt{b}$, $T_2 = \texttt{ac}$, $T_3 = \texttt{abb}$, and $T_4 = \texttt{abb}$.

\begin{figure}
	\centering{\includegraphics[width=\linewidth]{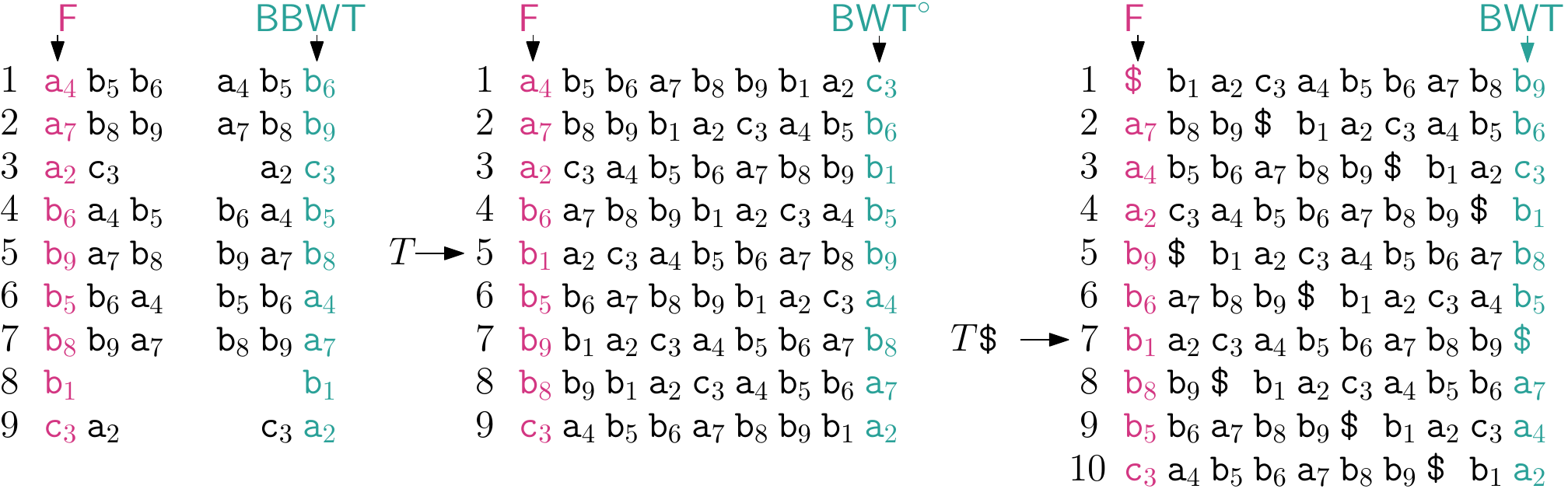}
	}\caption{All three BWT variants studied in this paper applied on our running example $T = \texttt{\RunningExample}$.
	\emph{Left}: \BBWT{} built on the last characters of the conjugates of all Lyndon words sorted in the \OmegaOrder{} order.
	\emph{Middle} and \emph{Right}: \BWTC{} and \BWT{} built on the lexicographically sorted conjugates of~$T$ and of~$T\texttt{\$}$, respectively.
	To ease understanding, each character is marked with its position in~$T$ in subscript.
	Reading these positions in $\First$ of \BBWT{} and in \First{} of \BWT{} gives a circular suffix array (there are multiple possibilities with $T_3 = T_4 = \texttt{abb}$) and the suffix array (the position of \texttt{\$} is uniquely defined as $|T\texttt{\$}| = 10$).
}
	\label{figBWTs}
\end{figure}

\subsection{Burrows-Wheeler Transforms}\label{secBWT}
We denote the bijective BWT of~$T$ by \BBWT{}, where $\BBWT[i]$ is the last character of the $i$-th string in the list 
storing the conjugates of all Lyndon factors $T_1,\ldots,T_t$ of~$T$ sorted with respect to the \OmegaOrder{} order. 
A property of \BBWT{} used in this paper as a starting point for an inversion algorithm is the following:

\begin{lemma}[{\cite[Lemma~15]{bonomo14sorting}}]\label{lemBBWTFirstChar}
	$\BBWT[1] = T[n]$.
\end{lemma}
\begin{proof}
	There is no conjugate of a Lyndon factor that is smaller than the smallest Lyndon factor~$T_t$ since 
	$T_t \LexOrderEq T_x \LexOrder T_x[j\twodots{}]$ for every $j \in [2\twodots{}|T_x|]$ and every $x \in [1\twodots{}t]$.
	Therefore, $T_t$ is the smallest string among all conjugates of all Lyndon factors.
	Hence, $\BBWT[1]$ is the last character of $T_t$, which is $T[n]$.
\end{proof}

The BWT of $T$, called in the following \BWT{}, is the BBWT of $\texttt{\$}T$ for a delimiter $\texttt{\$} \not\in \Sigma$ smaller than all other characters in $T$ (cf.~\cite[Lemma~12]{giancarlo07principles} since $\texttt{\$}T$ is a Lyndon word).
Originally, the BWT is defined by reading the last characters of all cyclic rotations of~$T$ (without \texttt{\$}) sorted lexicographically~\cite{burrows94bwt}.
Here, we call the resulting string \BWTC{}. 
\BWTC{} is equivalent to \BWT{} if $T$ contains the aforementioned unique delimiter~\texttt{\$}.
We further write $\BWT_P$ (and analogously $\BBWT_P$ or $\BWTC_P$) to denote the BWT of $P$ for a string~$P$.

Since $\BWT$ (and analogously $\BBWT$ or $\BWTC$) is a permutation of~$T$, 
it is natural to identify each entry of~$\BWT$ with a text position:
By construction $\BWT[i] = T[j]$, where $T[j+1\twodots{}]$ is the $i$-th lexicographically smallest suffix, i.e., $\SA[i] = j+1$, where $\SA$ is the suffix array of $T$.
A similar relation is given between $\BBWT$ and the circular suffix array~\cite{hon12ebwt,bannai19bbwtconstarxiv}, which is 
uniquely defined up to positions of equal Lyndon factors.
\Cref{figBWTs} gives an example for all three variants.
In what follows, we review means to simulate a linear traversal of the text in forward or backward manner by \BWT{}, and then translate this result to \BBWT{}.

\subsection{Backward and Forward Steps}\label{secSearchSteps}
Having the location of $T[i]$ in \BWT{},
we can compute $T[i+1]$ (i.e., $T[1]$ for $i=1$) and $T[i-1]$ (i.e., $T[n]$ for $i=0$) by rank and select queries.
To move from~$T[i]$ to $T[i+1]$, which we call a \teigi{forward step}, we can use the \FL{} mapping:
\begin{equation}\label{eqForwardSearch}
	\FL[i] := \BWT.\select_{\First[i]}( \First.\rank_{\First[i]}(i)),
\end{equation}
where $\First[i]$ is the $i$-th lexicographically smallest character in \BWT{}.
To move from~$T[i]$ to $T[i-1]$, we can use the \teigi{backward step} of the FM-index~\cite{ferragina00fmindex},
which is also called \LF{} mapping, and is defined as follows:
\begin{equation}\label{eqBackwardSearch}
	\LF[i] := \First.\select_{\BWT[i]}( \BWT.\rank_{\BWT[i]}(i))
	= C[\BWT[i]] + \BWT.\rank_{\BWT[i]}(i),
\end{equation}
where $C[c]$ is the number of occurrences of those characters in~$\BWT$ that are smaller than $c$ (for each character~$c \in [1\twodots{}\sigma]$).
We observe from the second equation of~(\ref{eqBackwardSearch}) that there is no need for~$\First$ when having~$C$.
This is important, as we can compute $C[i]$ in \Oh{n} time only having \BWT{} available.
Hence, we can compute $\LF[i]$ in \Oh{n} time in-place.
However, the same trick does not work with $\FL[i] = \BWT.\select_{\First[i]}(i - C[\First[i]])$.
To lookup $\First[i]$, we can use the selection algorithm of {Chan et al.}~\citet{chan18restore}
using \BWT{} and \Oh{\lg n} bits as working space (the algorithm restores \BWT{} after execution) to compute an entry of~$\First$ in \Oh{n} time.

In summary, we can compute both $\FL[i]$ and $\LF[i]$ in-place in \Oh{n} time.
The algorithm of {Crochemore et al.}~\citet[Thm.~2]{crochemore15bwt} inverting \BWT{} in-place in \Oh{n^{2+\epsilon}} time 
uses the result of {Munro and Raman}~\citet{munro96selection} computing $\First[i]$ in \Oh{n^{1+\epsilon}} time for a constant $\epsilon > 0$ in the comparison model.
As noted by {Chan et al.}~\citet[Sect.~1]{chan18restore}, the time bound for the inversion can be improved to \Oh{n^2} time in the RAM model under the assumption that \BWT{} is rewritable.

If we allow more space, it is still
advantageous to favor storing $C$ instead of $\First$ if $\sigma = \oh{n}$ because storing $\First$ and $C$ in their plain forms take $n \lg \sigma$ bits and $\sigma \lg n$ bits, respectively.
To compute $\FL[i]$, we can also compute \FL{} without~$\First$ by endowing~$C$ with a predecessor data structure (which we do in \cref{secToRLBBWT}).

Finally, we also need \LF{} and \FL{} on \BBWT{} for our conversion algorithms.
We can define \LF{} and \FL{} similarly for \BBWT{} with the following peculiarity:

\noindent\begin{minipage}{0.85\linewidth}
\subsection{Steps in the Bijective BWT}\label{secRewinding}
The major difference to the BWT is that the LF mapping of the BBWT can contain multiple cycles, meaning 
that \LF{} (or \FL{}) recursively applied to a \BBWT{} position would result 
in searching circular (more precisely, the search stays within the same Lyndon factor).
This is because \BBWT{} is the extended BWT~\cite[Thm.~20 and Remark~12]{mantaci07ebwt} applied to the multiset of Lyndon factors~$\{T_1, \ldots, T_t\}$.
This fact was exploited for circular pattern matching~\cite{hon12ebwt}, but is not of interest here.
\end{minipage}
\hfill
\begin{minipage}{0.13\linewidth}
		\includegraphics[width=\ImgWidth]{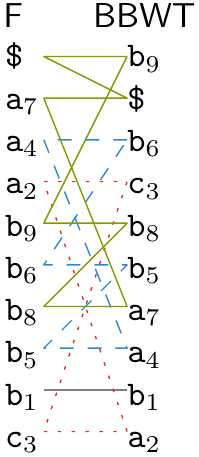}
\end{minipage}

Instead, we follow the analysis of the so-called \emph{rewindings}~\cite[Sect.~3]{bannai19bbwt}:
Remembering that we store the last character of all conjugates of all Lyndon factors in \BBWT{},
we observe that the entries in \BBWT{} representing the Lyndon factors (i.e., the last characters of the Lyndon factors)
are in sorted order (starting with $T_t[|T_t|]$ and ending with $T_1[|T_1|]$).
That is because the lexicographic order and the \OmegaOrder{} order are the same for Lyndon words~\cite[Thm.~8]{bonomo14sorting}.
Applying the backward step at such an entry results in a rewinding, i.e., 
we can move from the beginning of a Lyndon factor~$T_x$ (represented by $T_x[|T_x|]$ in \BBWT{})
to the end of~$T_x$ (represented by $T_x[1]$ in \BBWT{}) with one backward step.
We use this property with \cref{lemBBWTFirstChar} in the following sections to read the Lyndon factors from~$T$ individually in the order $T_t, \ldots, T_1$.

\section{Run-Length Compressed Conversions}\label{secRunLengthCompressed}

We now consider \BWT{} and \BBWT{} represented as run-length compressed strings taking
\Oh{\runBWT \lg n} and \Oh{\runBBWT \lg n} bits of space,
where \runBWT{} and \runBBWT{} are the number of character runs in \BWT{} and \BBWT{}, respectively.
For $r := \max(\runBWT, \runBBWT)$, the goal of this section is the following:

\begin{theorem}\label{thmRunLength}
We can convert~\RLBBWT{} to~\RLBWT{} in \Oh{n \lg r / \lg \lg r} time using \Oh{r \lg n} bits as working space, or vice versa.
\end{theorem}

To prove this theorem, we need a data structure that works in the run-length compressed space 
while supporting rank and select queries as well as updates more efficiently than the \Oh{n} time in-place approach described in \cref{secSearchSteps,secRewinding}:

\subsection{Run-length Compressed Wavelet Trees}\label{secRLwavelet}
Given a run-length compressed string~$S$ of uncompressed length~$n$ with $r$ character runs, there is an \Oh{r \lg n} bits representation of~$S$
that supports access, rank, select, insertions, and deletions in \Oh{\lg r} time~\cite[Lemma~1]{policriti18lz77}.
It consists of (1) a dynamic wavelet tree maintaining the starting characters of each character run and (2) a dynamic Fenwick tree maintaining the 
lengths of the runs.
It can be accelerated to \Oh{\lg r / \lg \lg r} time by using the following representations:
\begin{enumerate}
	\item The dynamic wavelet tree of {Navarro and Nekrich}~\citet{navarro14dynamic} on a text of length~$r$ uses \Oh{r \lg r} bits, and supports both updates and queries in \Oh{\lg r / \lg \lg r} time.
	\item The dynamic Fenwick tree of Bille et al~\cite[Thm.~2]{bille18partial} on $r$ $(\lg n)$-bit numbers uses \Oh{r \lg n} bits, 
		and supports both updates and queries in constant time if updates are restricted to be in-/decremental.
\end{enumerate}
The obtained time complexity of this data structure directly improves the construction of \RLBWT{}:
\begin{corollary}[{\cite[Thm.~2]{policriti18lz77}}] \label{corTextToRLBWT}
	We can construct the \RLBWT{} in \Oh{\runBWT \lg n} bits of space online on the reversed text in \Oh{n \lg \runBWT / \lg \lg \runBWT} time.
\end{corollary}

In the run-length compressed wavelet tree representation,
\RLBWT{} and \RLBBWT{} support an update operation and a backward step in \Oh{\lg r / \lg \lg r} time with $r := \max(\runBWT, \runBBWT)$.
This helps us to devise the following two conversions:

\subsection{From \RLBBWT{} to \RLBWT{}}\label{secToRLBWT}
We aim for directly outputting the characters of~$T$ in reversed order
since we can then use the algorithm of \cref{corTextToRLBWT} building \RLBWT{} online on the reversed text.
We start with the first entry of \BBWT{} (corresponding to the last Lyndon factor~$T_t$, i.e., storing $T_t[|T_t|] = T[n]$ according to \cref{lemBBWTFirstChar}) and do a backward step until we come back at this first entry (i.e., we have visited all characters of~$T_t$).
During that search, we copy the read characters to \RLBWT{} and mark in an array~$R$ of length~$\runBBWT$ at entry~$i$ how often we visited the $i$-th character run of \RLBBWT{}.
Finally, we remove the read cycle of \RLBBWT{} by decreasing the run lengths of \RLBBWT{} by the numbers stored in~$R$.
By doing so, we remove the last Lyndon factor~$T_t$ from \RLBBWT{} and consequently know that the currently first entry of \BBWT{} must correspond to $T_{t-1}$.
This means that we can apply the algorithm recursively on the remaining \RLBBWT{} to extract and delete the Lyndon factors in reversed order 
while building \RLBWT{} in the meantime.
By removing~$T_t$, \BBWT{} is still a valid BBWT since \BBWT{} becomes the BBWT of $T' := T_1 \cdots T_{t-1}$ whose Lyndon factors are the same as of~$T$ (but without $T_t$).
Note that it is also possible to build \RLBWT{} in forward order, i.e., computing $\RLBWT_{T_1 \cdots T_x}$ for increasing $x$
by applying the algorithm of {Mantaci et al.}~\cite[Fig.~1]{mantaci14sa} while omitting the suffix array construction.

\subsection{From \RLBWT{} to \RLBBWT{}}\label{secToRLBBWT}
To build \BBWT{}, we need to be aware of the Lyndon factors of~$T$, which we compute with \cref{lemLyndonFactorizationLinearTime} by simulating a forward scan on~$T$ with \FL{} on \BWT{}.
To this end, we store the entries of the $C$ array in a Fusion tree~\cite{fredman93fusion} using \Oh{\sigma \lg n} bits and supporting predecessor search in $\Oh{\lg \sigma / \lg \lg \sigma} = \Oh{\lg r / \lg \lg r}$ time.\footnote{We assume that the alphabet~$\Sigma$ is \emph{effective}, i.e., that each character of~$\Sigma$ appears at least once in~$T$. Otherwise, assume that $T$ uses $\sigma'$ characters.
	Then we build the static dictionary of {Hagerup}~\citet{hagerup99dictionary} in \Oh{\sigma' \lg \sigma'} time, supporting access to a character in $\Oh{\lg \lg \sigma'} = \Oh{\lg \lg r}$ time, assigning each of the $\sigma'$ characters an integer from $[1\twodots{}\sigma']$. 
	We further map $\RLBWT$ to the alphabet $[1\twodots{}\sigma']$, which can be done in \Oh{r} time by using \Oh{r \lg n} space for a linear-time integer sorting algorithm.
}
This time complexity also covers a forward step in \RLBWT{} by simulating~$\First$ with the Fusion tree on~$C$.
Hence, this fusion tree allows us to apply \cref{lemLyndonFactorizationLinearTime} computing the Lyndon factorization of~$T$
with a multiplicative \Oh{\lg r / \lg \lg r} time penalty
since this algorithm only needs to perform forward traversals.
The starting point of such a traversal is the position~$i$ with $\BWT[i] = \texttt{\$}$ because $\FL[i]$ returns the first character of $T$.
Whenever we detect a Lyndon factor~$T_x$ (starting with $x=1$), we copy this factor to our dynamic \RLBBWT{}.
For that, we always maintain the first and the last position of~$T_x$ in memory.
Having the last position of~$T_x$, we perform backward steps on \RLBWT{} until returning at the first position of~$T_x$ to read the characters of~$T_x$ in reversed order.
Then we continue with the algorithm of \cref{lemLyndonFactorizationLinearTime} at the position after~$T_x$ (for recursing on $T_{x+1}$).
Inserting a Lyndon factor into \RLBBWT{} works exactly as sketched by {Bonomo et al.}~\citet[Thm.~17]{bonomo14sorting} or in \cref{algoConstructBBWT} in the appendix (we review this algorithm in detail in \cref{secConstructBWTC}).

\section{In-Place Conversions}\label{secInPlaceConversions}
We finally present our in-place conversions that work in quadratic time by computing \LF{} or \FL{} in \Oh{n} time having only stored either \BWT{}, \BBWT{}, or \BWTC{}.
We note that the constructions from the text also work in the comparison model, 
while inverting a transform or converting two different transforms have a multiplicative \Oh{n^\epsilon} time penalty as the fastest option to access \First{} in the comparison model uses \Oh{n^{1+\epsilon}} time for a constant $\epsilon > 0$~\cite{munro96selection}.
We start with the construction and inversion of \BWTC{} (\cref{secConstructBWTC,secInvertBWTC}), 
where we show (a) that we can construct \BWTC{} from the text in the same manner as {Bonomo et al.}~\citet{bonomo14sorting} construct \BBWT{},
and (b) that the latter construction works also in-place.
Next, we show in \cref{secInvertBBWT} how to invert~\BBWT{} with the BWT inversion algorithm of {Crochemore et al.}~\citet[Fig.~3]{crochemore15bwt}, which allows us to also convert \BBWT{} to \BWT{} with the BWT construction algorithm of the same paper~\cite[Fig.~2]{crochemore15bwt}.
Finally, we show a conversion from \BWT{} to \BBWT{} in \cref{secBWTtoBBWT}.
An overview is given in \cref{tableOverviewInPlace}.

\begin{table}
	\centerline{\begin{tabular}{rrrrr}
			\toprule
			\diagbox{To}{From} & $T$ & $\BWT$ & $\BBWT$ & $\BWTC$ \\
			\midrule
			$T$     & \textbackslash              & \cite[Fig.~3]{crochemore15bwt}  & \cref{secInvertBBWT}  & \cref{secInvertBWTC}  \\
			$\BWT$  & \cite[Fig.~2]{crochemore15bwt}  & \textbackslash             & \cref{secInvertBBWT}  & \\
			$\BBWT$ & \cref{secConstructBWTC}      & \cref{secBWTtoBBWT}  & \textbackslash             & \\
			$\BWTC$ & \cref{secConstructBWTC}         &                            &                     & \textbackslash
			\\\bottomrule
		\end{tabular}
	}\caption{Overview of in-place conversions in focus of \cref{secInPlaceConversions} working in quadratic time.
}
	\label{tableOverviewInPlace}
\end{table}

\subsection{Constructing \BWTC{} and \BBWT{}}\label{secConstructBWTC}
We can compute \BWTC{} and \BBWT{} from~$T$ with the algorithm of {Bonomo et al.}~\citet{bonomo14sorting} computing the extended BWT~\cite{mantaci07ebwt}.
The extended BWT is the BWT defined on a set of primitive strings. 
As stated in \cref{secRewinding}, the extended BWT coincides with \BBWT{} if this set of primitive strings is the set of Lyndon factors of~$T$~\cite[Thm.~14]{bonomo14sorting}.
We briefly describe the algorithm of {Bonomo et al.}~\citet{bonomo14sorting} for computing the BBWT (cf.~\cref{figConstructBBWT}\ and \cref{algoConstructBBWT} in the appendix):
For each Lyndon factor~$T_x$ (starting with $x = 1$), prepend $T_x[|T_x|]$ to \BBWT{}.
To insert the remaining characters of the factor~$T_x$, let $p \gets 1$ be the position of the currently inserted character.
Then perform, for each $j = |T_x| - 1$ down to $1$, a backward step $p \gets \LF[p] + 1$, and insert $T_x[j]$ at $\BBWT[p]$\ (cf.~\cref{algoConstructBBWT} in the appendix).
To understand why this computes \BBWT{}, 
we observe that the last character of the most recently inserted Lyndon factor~$T_x$ is always the first character in $\BBWT_{T_1 \cdots T_x}$ 
according to \cref{lemBBWTFirstChar}.
By recursively inserting the preceding character at the place returned by a backward step, 
we precisely insert this character at the position where we would expect it (another backward step from the same position~$p$ would then return the inserted character).
Using only $n$ backward steps and $n$ insertions, this algorithm works in-place in \Oh{n^2} time by simulating \LF{} as described in \cref{secSearchSteps}.

Consequently, we can build \BWTC{} if $T$ is a Lyndon word since in this case \BWTC{} and \BBWT{} coincide~\cite[Lemma~12]{giancarlo07principles}.
That is because sorting the suffixes of~$T$ is equivalent to sorting the conjugates of $T$ (if $T$ is a Lyndon word, then its Lyndon factorization consists only of $T$ itself).

It is easy to generalize this to work for a general string~$T$.
First, if $T$ is primitive, then we compute its so-called \teigi{Lyndon conjugate},
i.e., a conjugate of~$T$ that is a Lyndon word.
(The Lyndon conjugate of~$T$ is uniquely defined if $T$ is primitive.)
We can find the Lyndon conjugate of~$T$ in \Oh{n} time with the following two lemmata:

\begin{lemma}[ {\cite[Prop.~1.3]{duval83lyndon}} ] \label{lemLyndonConcat}
	Given two Lyndon words~$S$ and~$T$, $ST$ is a Lyndon word if $S \LexOrder T$.
\end{lemma}

\begin{lemma}\label{lemLyndonConjugate}
	Given a primitive string~$T$, we can find its Lyndon conjugate in \Oh{n} time with \Oh{\lg n} bits of space.
\end{lemma}
\begin{proof}
	We use \cref{lemLyndonFactorizationLinearTime} to detect the last Lyndon factor~$T_t$ of
	the Lyndon factorization~$T_1 \cdots T_t$ of~$T$ with \Oh{\lg n} bits of working space.
	According to \cref{lemLyndonConcat}, $T_t T_1$ is a Lyndon word since $T_t \LexOrder T_1$, 
	and so is $T_t T_1 \cdots T_{t-1}$ a Lyndon word by a recursive argument.
	Hence, we have found $T$'s Lyndon conjugate.
\end{proof}

Let $\conj[j]{T}$ be the Lyndon conjugate of~$T$ for $j \in [0\twodots{}n-1]$.
Since \BWTC{} is identical to $\BBWT_{\conj[j]{T}}$, we are done by running the algorithm of {Bonomo et al.}~\citet{bonomo14sorting} on $\conj[j]{T}$.
Finally, if $T$ is not primitive, then there is a primitive string~$P$ such that $T = P^k$ for an integer $k \ge 2$.
We can compute $\BWTC_P$ with the above considerations.
For obtaining $\BWTC$, according to~\cite[Prop.~2]{mantaci03sturmian}, we only need to make each character in $\BWTC_P$ to a character run of length~$k$, 
i.e., if $\BWTC_P[i] = \texttt{c}$, we append $\texttt{c}^k$ to \BWTC{} for increasing $i \in [1\twodots{}|P|]$~(cf.~\cite[Thm.~13]{giancarlo07principles}).
Checking whether $T$ is primitive can be done in \Oh{n^2} time by checking for each pair of positions their longest common prefix.
We summarized these steps in the pseudo code of \cref{algoConstructBWTC} in the appendix.

\begin{figure}[t]
	\centering{\includegraphics[width=\ImgWidth]{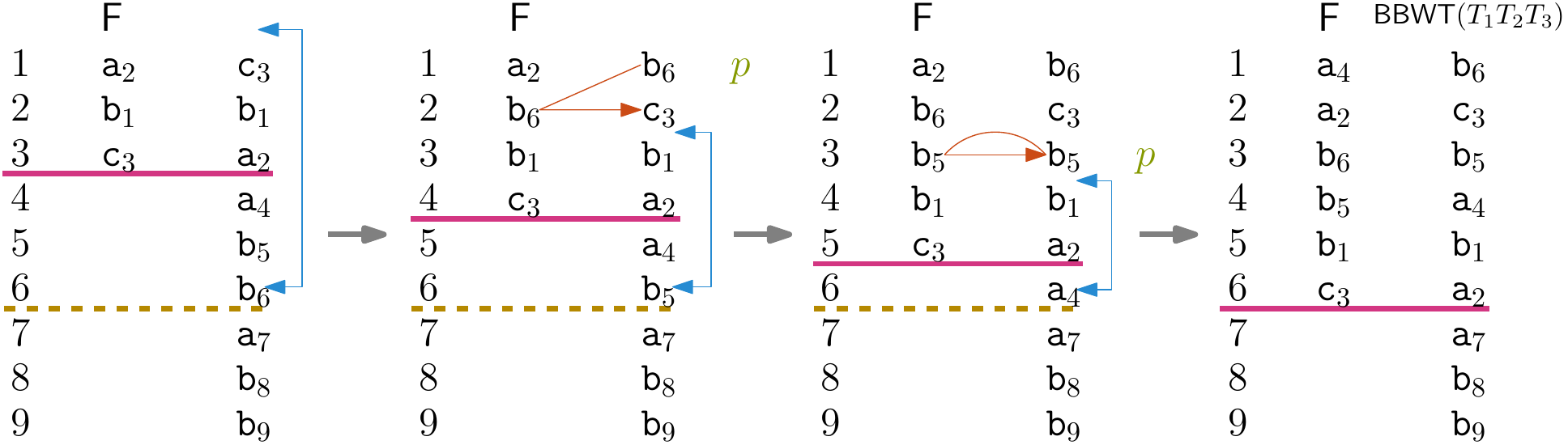}
	}\caption{Computing \BBWT{} from our running example $T = \texttt{\RunningExample}$ in four steps (visualized by four columns separated by three arrows \IroHako{color=gray,-{latex[round]}}), cf.~\cref{secConstructBWTC}.
		In each column, the characters from the top to the solid horizontal line (\IroHako{color=solarizedMagenta,dash pattern=on 10pt off 0pt}) form the currently built BBWT\@.
		The characters below that up to the dashed horizontal line (\IroHako{color=solarizedYellow,dash pattern=on 10pt off 10pt}) are under consideration of being merged into BBWT\@.
		This dashed line is always before the beginning of the next yet unread Lyndon factor.
		\emph{First column}: We have already computed the BBWT of $T_1 T_2 = \texttt{bac}$, which is \texttt{cba}.
		In the following we want to add the next Lyndon factor~$T_3 = \texttt{abb}$ to it.
		For that, we prepend its last character to the currently constructed BBWT.
		\emph{Second column}: We move the last character above the dashed line to the position ${\color{solarizedOrange}\LF}[{\color{solarizedGreen}p}] + 1$ with ${\color{solarizedGreen}p} = 1$, and update ${\color{solarizedGreen}p} \gets {\color{solarizedOrange}\LF}[{\color{solarizedGreen}p}]+1$.
		We recurse in the \emph{third column}, and have produced the BBWT of $T_1 T_2 T_3 = \texttt{bacabb}$ in the \emph{forth column}.
}
\label{figConstructBBWT}
\end{figure}

\subsection{Inverting \BWTC{}}\label{secInvertBWTC}
To invert~$\BWTC$, we use the techniques of {Crochemore et al.}~\citet[Fig.~3]{crochemore15bwt} inverting~\BWT{} in-place in \Oh{n^2}~time.
An invariant is that the BWT entry, whose \FL{} mapping corresponds to the next character to output, is marked with a unique delimiter \texttt{\$}.
Given that $\BWT[i] = \texttt{\$}$, 
the algorithm outputs $\BWT[\FL[i]]$, 
sets $\BWT[\FL[i]] \gets \texttt{\$}$, 
removes $\BWT[i]$, 
and recurses until \texttt{\$} is the last character remaining in \BWT{}.
By doing so, it restores the text in text order.

To adapt this algorithm for inverting~\BWTC{}, we additionally need a pointer~$p$ storing the first symbol of the text (since there is no unique delimiter such as \texttt{\$} in general).
Given that $p$ points to $\BWTC[i]$, 
we set~$i \gets \FL[i]$ and subsequently output $\BWTC[i]$.
From now on, the algorithm works exactly as \cite[Fig.~3]{crochemore15bwt} if we set $\BWTC[i] \gets \texttt{\$}$ after outputting $\BWTC[i]$\ (cf.~\cref{algoInvertBWTC} in the appendix).
More involving is inverting \BBWT{} or converting \BBWT{} to \BWT{}, which we tackle next.

\begin{figure}[t]
	\centering{\includegraphics[width=\ImgWidth]{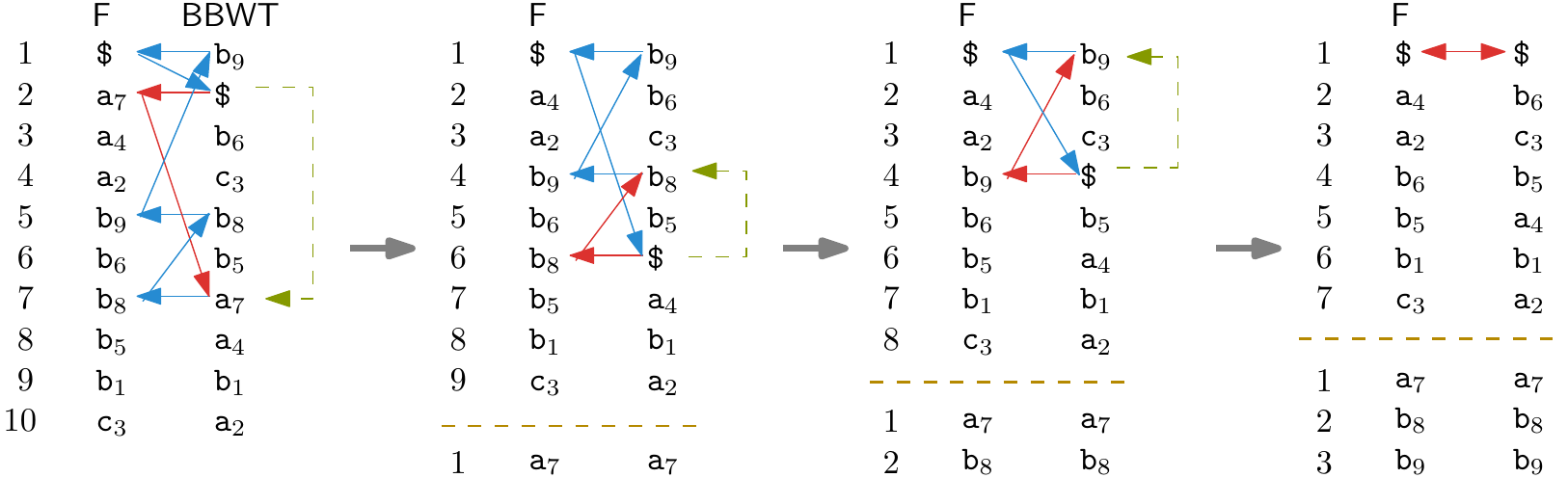}
	}\caption{Inverting \BBWT{} of our running example $T = \texttt{\RunningExample}$ (cf. \cref{secInvertBBWT}).
		\emph{First Column}: We prepend the \texttt{\$} delimiter to the last Lyndon factor~$T_t$ by inserting \texttt{\$} at $\BBWT[2]$.
		A {\color{solarizedRed}forward step} symbolized by the dashed arrow (\IroHako{color=solarizedGreen,dash pattern=on 5pt off 5pt,-{latex},line width=0.3em})
		leads us from \texttt{\$} to the first character of $T_t$.
		\emph{Second Column}: We output $\BBWT[6] = T_t[1] = T[7]$, remove \texttt{\$} and update $\BBWT[6] \gets \texttt{\$}$.
		The output is appended to the string shown below the dashed horizontal line
		(\IroHako{color=solarizedYellow,dash pattern=on 15pt off 10pt,line width=0.4em}).
		We continue with a {\color{solarizedRed}forward step} to access $\BBWT[4] = T_t[2] = T[8]$, and recurse in the \emph{third column}.
		\emph{Forth Column}: Since a {\color{solarizedRed}forward step} returns \texttt{\$}, we know that we have successfully extracted~$T_t = \texttt{abb}$.
}
\label{figInvertBBWT}
\end{figure}

\subsection{Inverting \BBWT{}}\label{secInvertBBWT}
Similarly to \cref{secToRLBWT}, we read the Lyndon factors from \BBWT{} in the order $T_t, \ldots, T_1$, and move each read Lyndon factor directly to a text buffer
such that
while reading the last Lyndon factor~$T_x$ for an $x \in [1\twodots{}t]$ from $\BBWT_{T_1 \cdots T_x}$, we move the characters of $T_x$ to $T_{x+1} \cdots T_t$, 
producing $\BBWT_{T_1 \cdots T_{x-1}}$ and $T_x \cdots T_t$.
This allows us to recurse by reading always the last Lyndon factor~$T_x$ stored in $\BBWT_{T_1 \cdots T_x}$.

Here, we want to apply the inversion algorithm for \BWTC{} described in \cref{secInvertBBWT}.
For adapting this algorithm to work with \BBWT{}, it suffices to insert \texttt{\$} at $\BBWT[2]$ (cf.~\cref{figInvertBBWT}).
By doing so, we add \texttt{\$} to the cycle of the currently last Lyndon factor~$T_x$ stored in \BBWT{}, i.e., we enlarge the Lyndon factor~$T_x$ to $\texttt{\$}T_x$.
That is because (a) $\BBWT[1]$ corresponds to the last character~$T_x[|T_x|]$ of~$T_x$ (cf.~\cref{lemBBWTFirstChar}), and after inserting \texttt{\$},
$\First[1] = \texttt{\$}, \First[2] = T_x[1]$, hence $\FL[1] = 2$ (a forward step on the last character of $T_x$ gives \texttt{\$}) and $\FL[2]$ gives the position in \BBWT{} corresponding to $T_x[1]$.
Moreover, inserting \texttt{\$} makes \BBWT{} the BBWT of $T' := T_1 \cdots T_{x_1} \texttt{\$} T_{x}$, where $\texttt{\$} T_x$ is the last Lyndon factor of~$T'$.
We now use the property that $\texttt{\$}T_x[i\twodots{}|T_x|]$ is a Lyndon word for each $i \in [1\twodots{}|T_x|]$, 
allowing us to perform the inversion steps of {Crochemore et al.}~\citet[Fig.~3]{crochemore15bwt} on~$\BBWT{}$.
By doing so, we can remove the entry of \BBWT{} corresponding to \conj[j]{T_x} for increasing $j \in [0\twodots{}|T_x|-1]$ and prepend the extracted characters to the text buffer storing~$T_{x+1} \cdots T_{t}$ within our working space while keeping \BBWT{} a valid BBWT.

Instead of inverting \BBWT{}, we can convert \BBWT{} to \BWT{} in-place by running the in-place BWT construction algorithm of {Crochemore et al.}~\citet[Fig.~2]{crochemore15bwt} on the text buffer after the extraction of each Lyndon factor. 
Unfortunately, this works not character-wise, but needs a Lyndon factor to be fully extracted before inserting its characters into \BWT{}.
Interestingly, for the other direction (from \BWT{} to \BBWT{}), we can propose a different kind of conversion that works directly on \BWT{} without decoding it.

\subsection{From \BWT{} to \BBWT{} on the Fly}\label{secBWTtoBBWT}

\newcommand*{\Pos}[1]{\ensuremath{i_{\texttt{#1}}}}

Like in \cref{secToRLBBWT}, we process the Lyndon factors of~$T$ individually to compute \BBWT{} by
scanning \BWT{} in text order to simulate \cref{lemLyndonFactorizationLinearTime}.
Suppose that we have built \BWT{} on $T\texttt{\$} \not= \texttt{\$}$ with
\texttt{\$} being the $(t+1)$-th Lyndon factor of $T\texttt{\$}$, 
and suppose that we have detected the first Lyndon factor~$T_1$.
Let \texttt{f} denote the last character of~$T_1$.\footnote{\texttt{f} stands for final character.}
Further let \Pos{f} and \Pos{\$} be the position of the last character of $T_1$ and the last character of $T$, respectively, 
such that $\BWT[\Pos{f}] = \texttt{f}$ and $\BWT[\Pos{\$}] = \texttt{\$}$.
Let $p := \LF[\Pos{f}]$ such that $\First[p] = \texttt{f}$ and $\BWT[p] = T_1[|T_1|-1]$ if $|T_1| > 1$ or $\BWT[p] = \texttt{\$}$ otherwise.
Since $T_1$ and $T_2$ are Lyndon factors, $T_1 \LexOrderReq T_2$.
Consequently, the suffix $T[\ibeg{T_2}\twodots{}]$ (the context of $\BWT[\Pos{f}]$) is lexicographically smaller than 
the suffix $T[\ibeg{T_1}\twodots{}]$ (the context of $\BWT[\Pos{\$}]$), i.e., $\Pos{f} < \Pos{\$}$.
\Cref{figBWTtoBBWTschematic} gives an overview of the introduced setting.

\begin{figure}[t]
	\begin{minipage}{0.5\linewidth}
		\includegraphics[width=\linewidth]{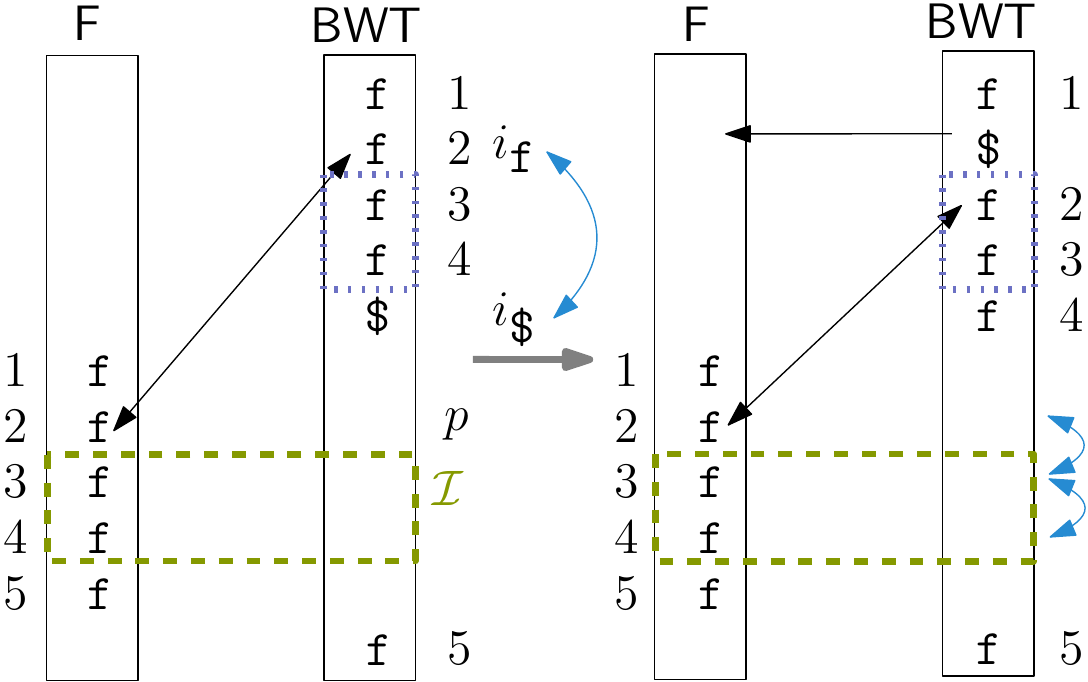}
	\end{minipage}
	\hfill
	\begin{minipage}{0.48\textwidth}
	\caption{Setting of \cref{secBWTtoBBWT} with focus on forming a cycle for a Lyndon factor ending with~\texttt{f} in \BWT{}. 
		\emph{Left}: We exchange $\BWT[\Pos{f}]$ with $\BWT[\Pos{\$}]$ with the aim to form a cycle.
		\emph{Right}: To obtain this cycle we additionally need to swap $\BWT[p]$ with the elements of
		the dashed rectangle~(\IroHako{color=solarizedGreen,dash pattern=on 10pt off 10pt}) corresponding to the interval~{\color{solarizedGreen}$\mathcal{I}$}
		having the same height as the dotted rectangle~(\IroHako{color=solarizedViolet,dash pattern=on 3pt off 6pt}) 
		covering~$\BWT[\Pos{f}+1\twodots{}\Pos{\$}-1]$.
}
\label{figBWTtoBBWTschematic}
\end{minipage}
\end{figure}

\begin{figure}[t]
	\centering{\includegraphics[width=\ImgWidth]{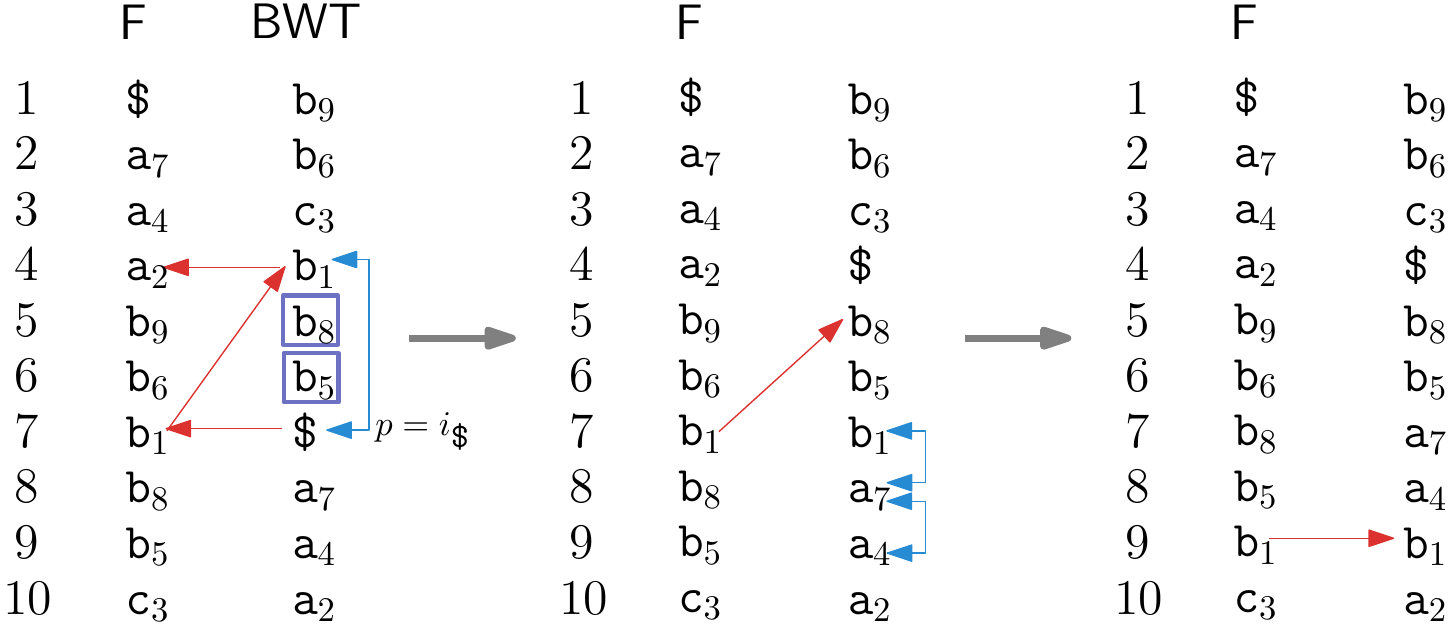}
	}\caption{Computing \BBWT{} from \BWT{} (cf. \cref{secBWTtoBBWT}) of our running example $T = \texttt{\RunningExample{}\$}$.  
		In the \emph{left} column, we find the first Lyndon factor~$T_1 = \texttt{b}$ of $T$ by {\color{solarizedRed}forward steps} with \FL{}.
		Since $|T_1| = 1$, $p = \Pos{\$}$.
		We obtain the \emph{middle} column by {\color{solarizedBlue}exchanging} $\BWT[4]$ with $\BWT[7]= \texttt{\$}$.
	Since there are two {\color{solarizedViolet}\texttt{b}}'s between \texttt{b} at $\BWT[4]$ and \texttt{\$} in the \emph{left} column,
	we need to {\color{solarizedBlue}swap} $\BWT[p]$ with the two elements below of it in the \emph{middle} column.
		This gives a cycle in the \emph{right} column.
		We can recurse since the \FL{} mapping of \texttt{\$} now yields the second character of~$T$.
}
\label{figBWTtoBBWTExample}
\end{figure}

\begin{figure}[t]
	\centering{\includegraphics[width=\ImgWidth]{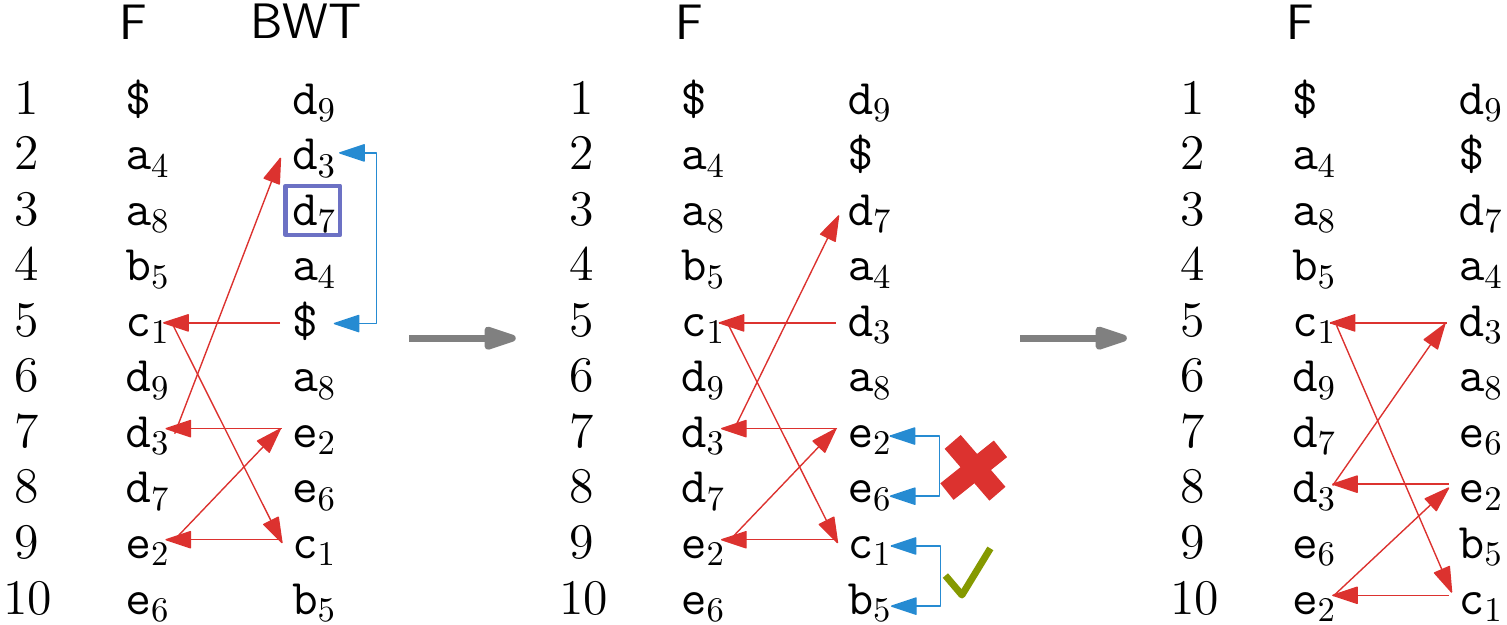}
	}\caption{Special case for computing \BWT{} from \BBWT{} (cf. \cref{secBWTtoBBWT}) with the different example string $T\texttt{\$} := \texttt{cedabedad\$}$ having $T_1 = \texttt{ced}$ as its first Lyndon factor.  
		\emph{Left column}: We find the first Lyndon factor~$T_1 = \texttt{ced}$ of $T$ by {\color{solarizedRed}forward steps} with \FL{}.
		Its last character is stored at $\BWT[2]$.
		By {\color{solarizedBlue}exchanging} \texttt{\$} with the last character of~$T_1$ in \BWT{}, we obtain the middle column.
		\emph{Middle column}: The \LF{} mapping for the third~{\color{solarizedViolet}\texttt{d}} in \First{} becomes invalid.
		However, there is only a character run of $T_1[|T_1|-1] = \texttt{e}$ in $\BWT{}$ of the $T_1[|T_1|] = \texttt{d}$ interval~$[7\twodots{}8]$ in \First{} starting with $p=7$.
		So we recurse on $\LF[p]$ to find 
		characters different from $T_1[|T_1|-2] = \texttt{c}$ to {\color{solarizedBlue}swap}
		in the respective $T_1[|T_1|-1] = \texttt{e}$ interval~$[9\twodots{}10]$. 
		\emph{Right Column}: We have created a cycle with the characters of the first Lyndon factor. 
		A forward step at \texttt{\$} gives the first character of the next Lyndon factor.
}
\label{figBWTtoBBWTSwapBeyond}
\end{figure}

Our aim is to change \BWT{} such that a forward or backward step within the characters belonging to $T_1$ always results in a cycle.
Informally, we want to cut $T_1$ out of \BWT{}, which additionally allows us to recursively continue with the \FL{} mapping to find the end of the next Lyndon factor~$T_2$.\footnote{As a matter of fact, if we now want to restore the text with the modified \BWT{} by \LF{}, we would only produce $T_2 \cdots T_t$.}
For that, we exchange $\BWT[\Pos{\$}]$ with $\BWT[\Pos{f}]$ (cf.~\cref{figBWTtoBBWTExample}).
Then the character $T[\iend{T_1}+1]$ (i.e., the first character of $T_2$) becomes the next character of $\texttt{\$}$ in terms of the forward step
($\BWT[\FL[\Pos{f}]] = T[\ibeg{T_2}]$),
while a backwards search on the first character of $T_1$ yields $T_1$'s last character (\LF{} returns \Pos{\$}, but now $\BWT[\Pos{\$}] = T_1[|T_1|] = \texttt{f}$).
This is sufficient as long as $\BWT[i] \not= \texttt{f}$ for every $i \in (\Pos{f}\twodots{}\Pos{\$}]$. 
Otherwise, it can happen that we change the mapping from the $i$-th \texttt{f} of $\First$ to the $i$-th \texttt{f} of $\BWT$ (or vice versa) unintentionally.
In such a case, we swap some entries in $\BWT$ within the \texttt{f} interval of $\First$.
In detail, we conduct the exchange ($\BWT[\Pos{\$}]$ with $\BWT[\Pos{f}]$), but continue with swapping $\BWT[i]$ and $\BWT[i+1]$ unless 
$\BWT[\FL[i]]$ becomes that \texttt{f} that corresponds to $T_1[|T_1|]$
for increasing $i$ starting with $i = p$ until $\First[i] \not= \texttt{f}$ or $\LF[i] \not\in [\Pos{f}\twodots{}\Pos{\$}]$.
This may not be sufficient if the characters we swap are identical (cf.~\cref{figBWTtoBBWTSwapBeyond}). 
In such a case, 
we recurse on the $T_1[|T_1|-1]$ interval of \First{}, see also \cref{algoBWTToBBWT} in the appendix.

Instead of checking whether we have created a cycle after each swap, we want to compute the exact number of swaps needed for this task.
For that we note that 
exchanging $\BWT[\Pos{\$}]$ with $\BWT[\Pos{f}]$ decrements the values of $\BWT.\rank_{\texttt{f}}(j)$ for every $j \in [\Pos{f}\twodots{}\Pos{\$}]$ by one.
In particular, $\BWT.\select_{\texttt{f}}$ changes for those \texttt{f}'s in \BWT{} that are between \Pos{f} and \Pos{\$}.
Hence, the number of swaps~$m$ is the number of positions~$k \in [\Pos{f}+1\twodots{}\Pos{\$}-1]$ 
with $\BWT[k] = \texttt{f}$.
The swaps are performed within the range~$\mathcal{I}$ starting with $p+1$ and covering all positions $i$ with $\LF[i] \in [\Pos{i}\twodots{}\Pos{\$}]$
and $\First[i] = \texttt{f}$ since $\mathcal{I}$ covers all entries whose mapping has changed.
However, if $\BWT[p\twodots{}]$ starts with a character run of $T[\iend{T_1}-1]$ (or of $T[\ibeg{T_1}]$ if $|T_1| = 1$)\footnote{For $|T_1| = 1$, $p = \Pos{\$}$, and hence, $\BWT[p]$ was \texttt{\$} but now is $\texttt{f} = T_1[|T_1|] = T_1[1]$.}, 
swapping the identical characters does not change \BWT{}, and therefore has no effect of changing \LF{}.
Instead, we search the end of this run within $\mathcal{I}$ to swap the first entry~$i$ below this run with the first entry of this run,
and recurse on swapping entry~$i$ with entries below of it.

\subparagraph{Correctness.}
To see why the swaps restore the LF mapping for $T_1$ and the remaining part of the text~$T_2 \cdots T_t$, 
we examine those substrings of~$T$ that we might no longer find with the LF mapping after exchanging $\BWT[\Pos{\$}]$ with $\BWT[\Pos{f}]$.

In detail, we examine each substring $S_j := x_j y_j \texttt{f} \in \Sigma^3$ with $j \in [1\twodots{}m]$
that is represented in \BWT{} (before changing it) with
$\BWT[p+j] = y_j, \BWT[\LF[p+j]] = x_j, \BWT[\FL[p+j]] = \texttt{f}$, 
and $i_j := \FL[p+j] \in [\Pos{f}+1\twodots{}\Pos{\$}-1]$.
Due to the LF-mapping, $\BWT.\select_{\texttt{f}}(\BWT.\rank_{\texttt{f}}(\Pos{f}) + j) = i_j$, 
meaning that $\BWT[i_j]$ is the $j$-th \texttt{f} in $\BBWT[\Pos{f}+1\twodots{}\Pos{\$}-1]$, which stores $m$ \texttt{f}'s.
After exchanging $\BWT[\Pos{\$}]$ with $\BWT[\Pos{f}]$, $\FL[p+j]$ becomes $i_{j+1}$ for $j \in [0\twodots{}m]$ with $i_{m+1} := \Pos{\$}$.
However, for all $i > p+m$, $\FL[i]$ did not change. 
Hence, we only have to focus on the range $\mathcal{I} = [p+1\twodots{}p+m]$.

First, suppose that $y_1 = \BWT[p+1] \not= \BWT[p]$.
If we swap $\BWT[p]$ with $\BWT[p+1]$, then $\LF[p]$ is still $i_1$, 
but $\BWT[\LF[p]]$ becomes $x_1$ such that we have fixed the substring~$x_1 y_1 \texttt{f}$. 
This also works in a more general setting:
If $y_j = \BWT[p+j] \not= \BWT[p]$ for every $j \in [1\twodots{}m]$, 
we can perform $m$ swaps like above for all $m$ entries in $\BWT[\mathcal{I}]$ to fix all substrings~$S_j$.

Now suppose that $y_j = \BWT[p+j] = \BWT[p]$ for $j \in [1\twodots{}\ell]$ with the largest possible $\ell \in [2\twodots{}m]$.
Let $k > p+\ell$ be the first entry with $\BWT[k] \not= \BWT[p]$.
First, suppose that $k \in \mathcal{I}$.
Then $\First[k] = \texttt{f}$, and swapping $\BWT[k]$ with $\BWT[p]$ restores the LF mapping for the substrings~$S_j$ with $j \in [1\twodots{}\ell]$
since this swap decrements $\BWT.\rank_{y_j}[p+j]$ by one for every $j \in [1\twodots{}\ell]$.
We recurse on swapping $\BWT[k]$ with the following \BBWT{} entries in~$\mathcal{I}$ until all $m$ substrings got restored.
Finally, if $k \ge p+m$, then all $y_i$ are equal such that we can find the $x_i$ in \BWT{} consecutively stored at positions with an \First{} value of $y_i$. Thus, we can apply the swaps there recursively.

\subparagraph{Time Complexity.}
Fixing a Lyndon factor~$T_x$, we spend \Oh{|\mathcal{I}|} time for the swaps in $\BWT[\mathcal{I}]$, 
and perform the swaps recursively at most $|T_x|$ times, where we need additionally \Oh{n} time per recursion step for computing $\LF[p]$, summing up to $\Oh{|T_x| (|\mathcal{I}| + n)} = \Oh{n|T_x|}$ time.
Since $\sum_{x=1}^t |T_x| = n$, we yield \Oh{n^2} total time.

\section{Open Problems}
Our algorithm of \cref{secInvertBBWT} converts \BBWT{} to \BWT{}, Lyndon factor by Lyndon factor.
It would be interesting to find another conversion that works character-wise.
Here, our inversion algorithm extracts a Lyndon factor in text order from \BBWT{}, while the used BWT construction algorithm parses the text in reverse text order.

{Crochemore et al.}~\citet[Sect.~4]{crochemore15bwt} proposed a space and time trade-off algorithm based on their in-place techniques computing or inverting~\BWT{}.
We are positive that it should be possible to adapt their techniques for computing or inverting~\BBWT{} or~\BWTC{} with a trade-off parameter.

From the combinatorial perspective, we question whether the number of distinct Lyndon words of~$T$ is bounded by the runs in \BBWT{}.
If we can affirm this question, it would be possible to adapt the \BBWT{} based index data structure~\cite{bannai19bbwt} for \RLBBWT{} using \Oh{\runBBWT{} \lg n} bits of space
because this solution needs a bit vector with rank and select support marking the positions in \BBWT{} corresponding to the distinct Lyndon factors. 
If this number is at most the number of runs~$\runBBWT$, 
then we can store this bit vector entropy-compressed in \Oh{r \lg n} bits when $\runBBWT = \oh{n}$ since 
$n H_0(r) = n \lg (n/(n-r)) + r \lg ((n-r)/r) \le n \lg r \Leftrightarrow r \lg ((n-r)/r) \le n \lg (r(n-r)/n)$ for $r = \runBBWT$.

Speaking of \RLBBWT{}, we wonder whether we can construct \RLBBWT{} online in run-length compressed space similar to \cref{corTextToRLBWT}. 
With the run-length compressed wavelet tree, the algorithm of {Bonomo et al.}~\citet[Thm.~17]{bonomo14sorting} works in \Oh{n \lg \runBBWT / \lg \lg \runBBWT} time with $\max_{x \in [1\twodots{}t]} |T_x| + \Oh{\runBBWT \lg n}$ bits of space by reading each Lyndon factor of the text individually.

\clearpage
\bibliographystyle{plainurl}
\bibliography{literature,references}

\clearpage
\appendix
\section{Computing the Transforms}
For readers unfamiliar with the three studied BWT variants introduced in \cref{secBWT}, 
we provide didactically simplified constructions in \cref{figBBWT,figBWT,figBWTC} that neglect the in-place restriction.

\begin{figure}[H]
	\centering{\includegraphics[width=\linewidth]{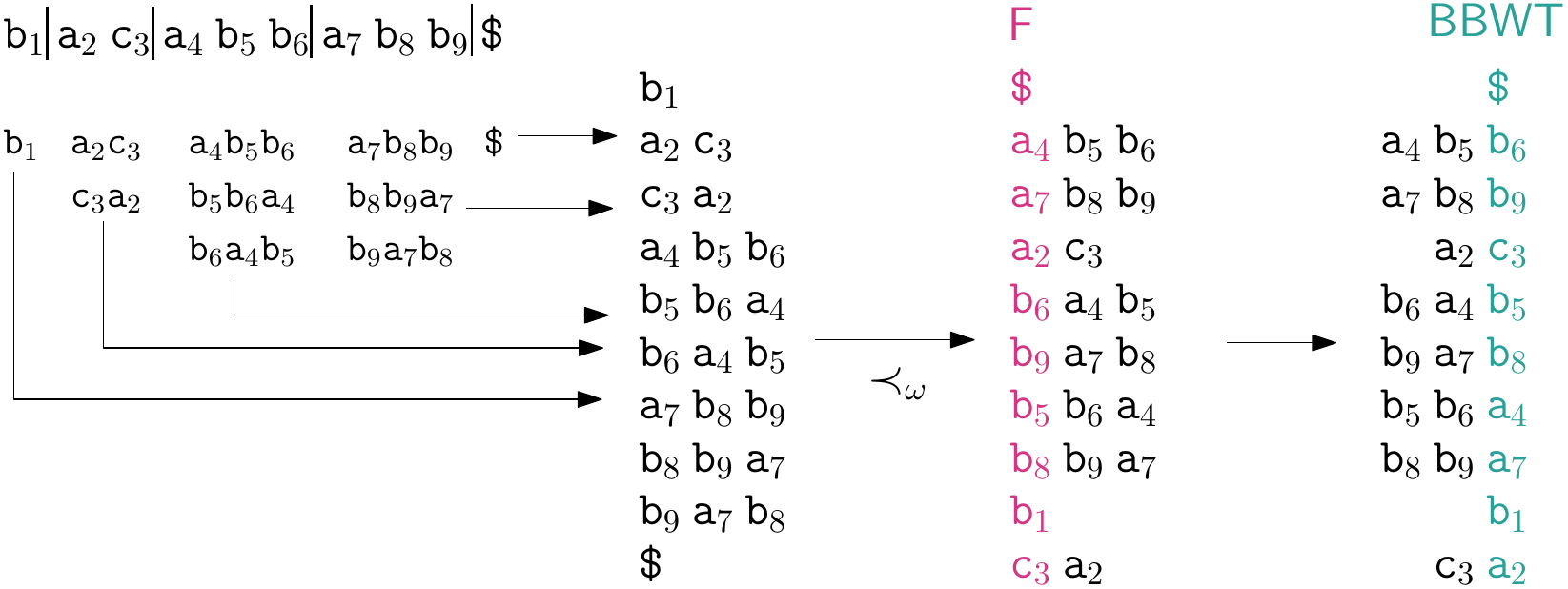}
	}\caption{Constructing \BBWT{} of $T\texttt{\$} = \texttt{\RunningExample}\texttt{\$}$.
	The Lyndon factorization of $T\texttt{\$}$ is visualized by the vertical bars.
	We take all conjugates of each Lyndon factor into a list, sort this list with respect to the $\OmegaOrder$ order.
	The first characters and the last characters in this list give \First{} and \BBWT{}, respectively.
}
\label{figBBWT}
\end{figure}

\begin{figure}[H]
	\centering{\includegraphics[width=\linewidth]{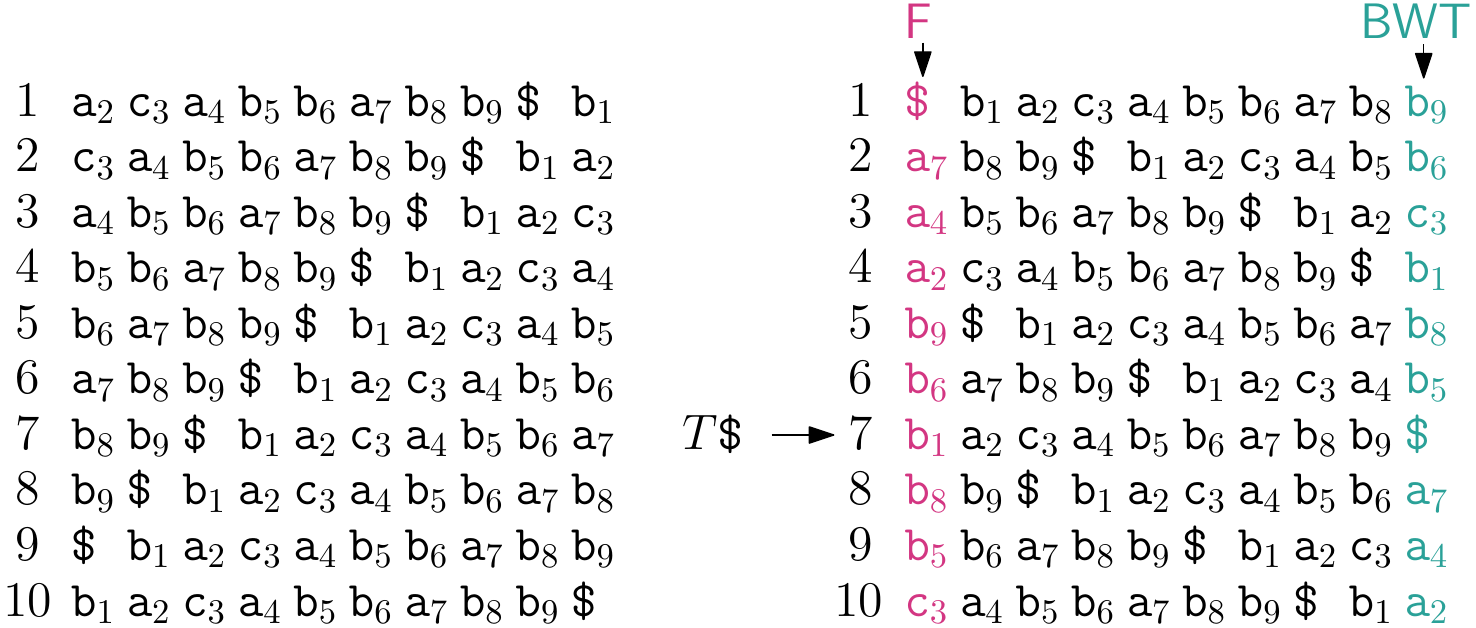}
	}\caption{Constructing \BWT{} of $\texttt{\$}T = \texttt{\$}\texttt{\RunningExample}$.
		Since $\texttt{\$}T$ is a Lyndon word, $\BWT_{\texttt{\$}T} = \BBWT_{\texttt{\$}T}$.
		To construct \BWT{}, we follow \cref{figBBWT}:
		The Lyndon factorization of $\texttt{\$}T$ consists only of~$\texttt{\$}T$ itself.
		Consequently, we take all conjugates of $\texttt{\$}T$ (\emph{left}) and sort them (\emph{right}).
		Thanks to the \texttt{\$}, it does not matter whether we sort by lexicographic order or $\OmegaOrder$ order~\cite[Lemma~7]{bonomo14sorting}.
		The first characters and the last characters in this sorted list give \First{} and \BWT{}, respectively.
	}
	\label{figBWT}
\end{figure}

\begin{figure}[H]
	\centering{\includegraphics[width=\linewidth]{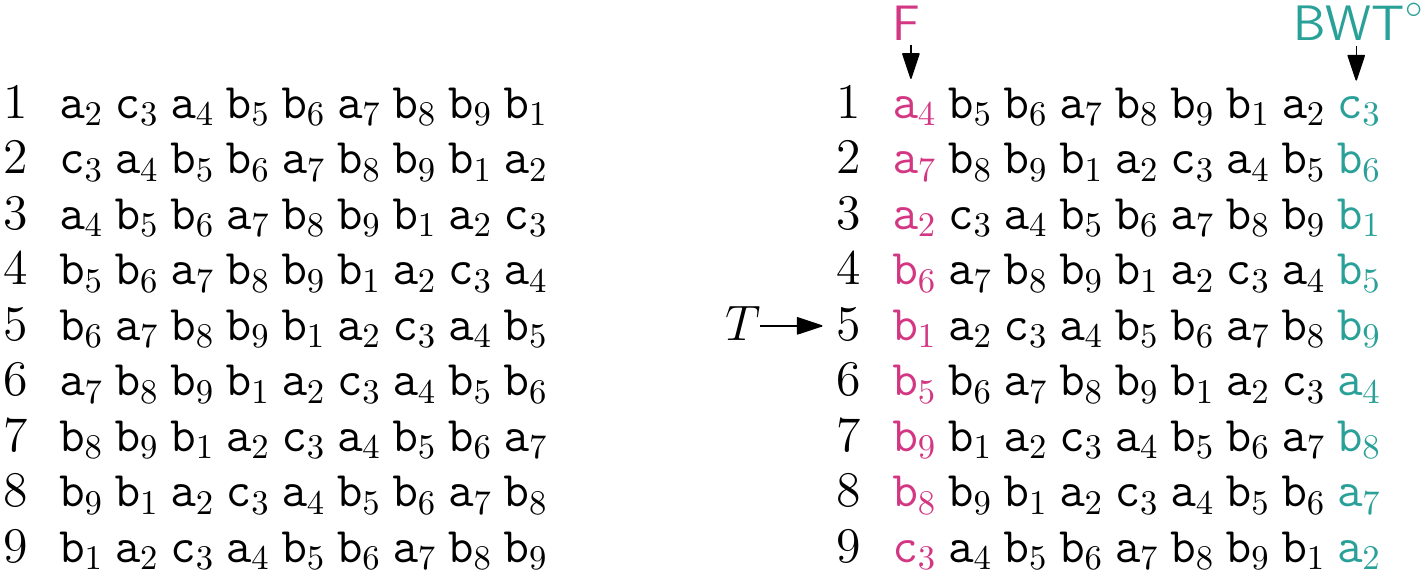}
	}\caption{Computing \BWTC{} of $T = \texttt{\RunningExample}$.
		Since $T$ is primitive, it has a Lyndon conjugate, which is $\conj[3]{T} = \texttt{abbabbbac}$.
		To construct \BWTC{}, we follow \cref{figBBWT}:
		We take all conjugates of $\conj[3]{T}$ (\emph{left}), and sort them lexicographically (\emph{right}).
		Here, the lexicographic order is equivalent to the $\OmegaOrder$ order~\cite[Lemma~7]{bonomo14sorting}.
		The first characters and the last characters in this list give \First{} and \BWTC{}, respectively.
	}
	\label{figBWTC}
\end{figure}

\section{From Run-Length Compressed Text to \RLBWT{}}
Complementary to the study of \cref{secRunLengthCompressed},
we want to examine how much space is needed to compute \RLBWT{} in $\runBWT \lg \sigma$ bits (without a wavelet tree) having the text also run-length compressed.
For that, we use the algorithm of {Crochemore et al.}~\citet[Fig.~2]{crochemore15bwt} to compute the \RLBWT{} in-place in \Oh{n^2} online on the reversed text.

\begin{lemma}
Given that \runText{} is the number of character runs in the run-length compressed text~$T$,
we can compute \RLBWT{} in \Oh{n^2} time with $\runBWT + \runText/2 + \Oh{1}$ words of working space, including the space of the run-length compressed text and \RLBWT{}.
\end{lemma}
\begin{proof}
Suppose $\texttt{b} := T[i] = \cdots = T[k]$ is a character run. 
Then adding this character run to \RLBWT{} induces at most two character runs in \BWT{}.
That is because (1) the context of $T[k]$ (i.e., the suffix $T[k+1\twodots{}]$ succeeding~$T[k]$) starting with $T[k+1] \cdots$ can be arbitrary, 
but (2) the context of $T[j]$ starts always with one or multiple \texttt{b}'s, for $j \in [i\twodots{}k)$. 
Hence, adding a character run of $T$ to \BWT{} can cause at most two new character runs in \BWT{}.
Consequently, it is possible that the run-length BWT computed on at least half of character runs of the text already needs as much space as $\RLBWT_T$.
In total, the memory peak is at most $\runBWT + \runText/2$ words, since we freed up at least half of character runs of $T$ when reaching this peak.
\end{proof}

\clearpage
\section{Pseudo Codes}

\begin{algorithm2e}
	\DontPrintSemicolon{}
	\caption{Duval's Algorithm, see \cref{lemLyndonFactorizationLinearTime} and \cite[Algo.~2.1]{duval83lyndon}} 
	\label{algoDuval}
	task: output the ending positions of all Lyndon factors \;
	$k \gets 0$ \Comment*{ending position of first Lyndon factor}
	\While{k < n}{$i \gets k+1$ and $j \gets k+2$ \;
		\While{$j \not= n+1$ and $T[i] \le T[j]$}{\lIf(\Comment*[f]{$T[k+1\twodots{}j]$ is a Lyndon word}){$T[i] < T[j]$}{$i \gets k+1$}
			\lIf(\Comment*[f]{$T[k+1\twodots{}i] = T[j-(i-k)\twodots{}j]$}){$T[i] = T[j]$}{$i \gets i+1$}
			$j \gets j+1$ \;
		}\Repeat(\Comment*[f]{invariant: $T[i] > T[j]$ or $j = n+1$}){$k \ge i$}{$k \gets k + (j-i)$  \;
			output $k$ 
}
	}
\end{algorithm2e}

\begin{algorithm2e}
	\DontPrintSemicolon{}
	\caption{Computing \BBWT{} from~$T$ \cite[Algo.~13]{bonomo14sorting}, cf.~\cref{secConstructBWTC}}
	\label{algoConstructBBWT}
	\ForEach{Lyndon factor $F_x$ with $x = 1$ up to $t$}{prepend $T_x[|T_x|]$ to \BBWT{} \;
		$p \gets 1$ \Comment*{insert position in \BBWT{}}
		\For{$i = |T_x| - 1$ down to $1$}{$p \gets \LF[p] + 1$ \;
			insert $T_x[i]$ right after $\BBWT[p-1]$ such that $\BBWT[p] = T_x[i]$
		}
	}
\end{algorithm2e}

\begin{algorithm2e}
	\DontPrintSemicolon{}
	\caption{Functions for computing $\LF[i] = C(\BWT, \BWT[i]) + \protect\rank(\BWT[1\twodots{}i], \BWT[i])$, cf.~\cref{eqBackwardSearch}}
	\label{algoLF}
	\SetKwProg{Function}{function}{:}{end}
	\begin{adjustbox}{valign=t}
	\begin{minipage}{0.45\linewidth}
	\RestyleAlgo{plain}
	\begin{function}[H]
	 \Function{$C(\BWT,c)$}{$\mathit{count} \leftarrow 0$; \;
		\For{$j=1$ up to $|\BWT|$}{\If{$\BWT[j] < c$}{$\mathit{count} \leftarrow \mathit{count} + 1$\;
			}
		}
		\Return $\mathit{count}$
	 }
\end{function}
	\end{minipage}
	\end{adjustbox}
	\begin{adjustbox}{valign=t}
	\begin{minipage}{0.45\linewidth}
	\RestyleAlgo{plain}
\begin{function}[H]
	 \Function{$\rank(\BWT,c)$}{$\mathit{rank} \leftarrow 0$\;
		\For{$j=1$ up to $|\BWT|$}{\If{$\BWT[j] = c$}{$\mathit{rank} \leftarrow \mathit{rank} + 1$\;
			}
		}
		\Return $\mathit{rank}$
	 }
\end{function}
	\end{minipage}
	\end{adjustbox}
\end{algorithm2e}

\begin{algorithm2e}
	\DontPrintSemicolon{}
	\caption{Computing \BWTC{} in-place for a Lyndon word~$T$, cf.~\cref{secConstructBWTC}}
	\label{algoConstructBWTC}
	let $\tilde{T}[1\twodots{}n]$ denote the working space of $n \lg \sigma$ bits storing $T[1\twodots{}n]$ \;
	swap $\tilde{T}[n]$ with $\tilde{T}[n-1]$ \Comment*{$\BWT_{T[n-1\twodots{}n]} = \tilde{T}[n-1\twodots{}n]$ since $T$ is Lyndon}
	$\mathit{lastchar} \leftarrow \tilde{T}[n]$ \Comment*{the character most recently inserted into $\BWTC$}
	$\mathit{lastpos} \leftarrow n$ \Comment*{the position of $\mathit{lastchar}$ in $\tilde{T}$}
	\For(\Comment*[f]{$\tilde{T}[1\twodots{}i-1] = T[1\twodots{}i-1]$ and $\tilde{T}[i\twodots{}n] = \BWTC[i\twodots{}n]$}){$i=n-1$ down to $2$}{$\mathit{lastpos} \gets i + C(\tilde{T}[i\twodots{}n], \mathit{lastchar}) + \rank(\tilde{T}[i\twodots{}\mathit{lastpos}], \mathit{lastchar})$
		\;
		\Comment{apply \cref{algoLF} for $\mathit{lastpos} \gets i + \LF[\mathit{lastpos}-i]$ with $\BWT[k] = \tilde{T}[k+i]~\forall k$}

$\mathit{lastchar} \leftarrow \tilde{T}[i-1] = T[i-1]$ \Comment*{save $\tilde{T}[i-1]$ before overwriting it.}
		\lFor{$j=i-1$ to $\mathit{lastpos}$}{$\tilde{T}[j] \gets \tilde{T}[j+1]$
		}
		$\tilde{T}[\mathit{lastpos}] \gets \mathit{lastchar}$
	}
	return $\tilde{T} = \BWTC$
\end{algorithm2e}

\begin{algorithm2e}
	\DontPrintSemicolon{}
	\caption{Inverting \BWTC{}, cf.~\cref{secInvertBWTC}}
	\label{algoInvertBWTC}
	$p \gets $ the position of~$T[1]$ in $\BWTC{}$ \;
	$T \gets$ empty string \;
	\While(\Comment*[f]{apply~\cite[Fig.~3]{crochemore15bwt}}){$|\BWTC| > 1$}{$p' \gets \FL[p]$ \;
		append $\BWTC[p']$ to $T$ \;
		$\BWTC[p'] \gets \texttt{\$}$ \;
		delete $\BWTC[p]$ \;
		$p \gets p'$ \;
	}
	return $T$\;
\end{algorithm2e}

\begin{algorithm2e}
	\DontPrintSemicolon{}
	\caption{Inverting \BBWT{}, cf.~\cref{secInvertBBWT}}
	\label{algoInvertBBWT}
	$T \gets$ empty string \;
	\While(\Comment*[f]{extract the currently last Lyndon factor}){\BBWT{} is not empty}{insert \texttt{\$} after $\BBWT[1]$ such that $\BBWT[2] = \texttt{\$}$ \;
		$p \gets 2$ \Comment*{$p$ is the position where \texttt{\$} is stored}

		\While(\Comment*[f]{apply~\cref{algoInvertBWTC}}){$p \not= 1$}{$p' \gets \FL[p]$ \;
			append $\BBWT[p']$ to the text buffer \;
			$\BBWT[p'] \gets \texttt{\$}$ \;
			delete $\BWT[p]$ \;
			$p \gets p'$ \;
		}
		delete $\BBWT[1]$ \Comment*{$p = 1$ and $\First[1] = \BBWT[1] = \texttt{\$}$}
	}prepend the text buffer to already restored text~$T$ \;
	return $T$\;
\end{algorithm2e}

\begin{algorithm2e}
		\newcommand*{\src}{\mathrm{src}}
		\newcommand*{\dst}{\mathrm{dst}}
		\newcommand*{\depth}{\mathrm{depth}}
		\DontPrintSemicolon{}
		\caption{Converting \BWT{} to~$\BBWT$ in-place, cf.~\cref{secBWTtoBBWT}}
		\label{algoBWTToBBWT}
		$x \gets 1$ \Comment*{counts the currently processed Lyndon factor~$T_x, x \in [1\twodots{}t]$}

		\While(\Comment*[f]{$\BWT[\FL[\Pos{\$}]] = T_x[1]$}){\textup{$\LF[\Pos{\$}] \not= \texttt{\$}$ where $\Pos{\$}$ is the position of \texttt{\$} in \BWT{}}}{$p \gets \FL^{(|T_x|-1)}[\Pos{\$}]$ \Comment*{apply \FL{} $(|T_x|-1)$ times}
			stipulate that $T_x[0] := \texttt{\$}$ \Comment*{$\BWT[p] = T_x[|T_x|-1]$ even for $|T_x|=1$}
			$\Pos{f} \gets \FL[p]$ \Comment*{$\First[p] = T_x[|T_x|], \BWT[\Pos{f}] = T_x[|T_x|]$}
		exchange $\BWT[\Pos{f}]$ with $\BWT[\Pos{\$}]$ \;
		$m \gets \left|\{ i \in [\Pos{f}+1\twodots{}\Pos{\$}-1] \text{~with~} \BWT[i] = T_x[|T_x|] \}\right|$ \;
		$\depth \gets 0$ \Comment*{recursion depth for the swaps}
		$p' \gets \LF[p]$ \Comment*{save $\LF[p]$ for the recursion}
		\While{$m > 0$}{$\dst \gets p, \src \gets p+1$ \Comment{try to swap $\BWT[\dst]$ with next unequal entry}
			\lWhile{$\BWT[\src] = T_x[|T_x|-\depth-1]$}{\label{lineIncrementSrc}
			$\src \gets \src + 1$}
			\If{$\First[\src] = T_x[|T_x|-\depth]$ and $\src \le p + m$}{swap~$\BWT[\src]$ with $\BWT[\dst]$ \;
				$\dst \gets \src$, $m \gets m - (\src - \dst)$ \;
				$p' \gets p' + (\src + \dst - 1)$ \Comment*{move $p'$ such that $\FL[p'] = p$}
				\textbf{goto} Line~\ref{lineIncrementSrc} \;
			}
			increment $\depth$ by one, $p \gets p'$, $p' \gets \LF[p']$ \;
		}
		$x \gets x+1$ \Comment*[f]{done with $T_1 \cdots T_{x-1}$}
}
invariant: $x = t$ \;
\end{algorithm2e}

\end{document}